\documentclass{article}
\usepackage[T1]{fontenc}

\usepackage{stackengine}
\usepackage{float}
\usepackage[utf8]{inputenc}
\usepackage{fullpage}
\usepackage{amsmath}
\usepackage{amssymb,amsfonts,dsfont}
\usepackage{amsthm}
\usepackage{thm-restate}
\usepackage{dsfont}
\usepackage{physics}
\usepackage{tikz}
\usetikzlibrary{calc,decorations.pathreplacing,fit,shapes.misc}
\usepackage{array}
\usepackage{comment}
\usepackage[backend=biber,style=alphabetic,maxbibnames=99,minalphanames=4,maxalphanames=5]{biblatex}
\usepackage{subcaption}

\usepackage{authblk}

\usepackage{todonotes}
\usepackage{lmodern}
\usetikzlibrary{decorations.pathmorphing}
\usepackage[colorlinks=true,linkcolor=blue,allcolors=blue]{hyperref}
\usepackage[capitalize,nameinlink]{cleveref}
\usepackage[ruled, vlined, linesnumbered]{algorithm2e}
\usepackage{setspace}
\usepackage{caption}

\usepackage{mathtools,bm}

\tikzset{snake it/.style={decorate, decoration=snake}}
\makeatletter
\tikzset{
  fitting node/.style={
    inner sep=0pt,
    fill=none,
    draw=none,
    reset transform,
    fit={(\pgf@pathminx,\pgf@pathminy) (\pgf@pathmaxx,\pgf@pathmaxy)}
  },
  reset transform/.code={\pgftransformreset}
}
\makeatother

\DeclareMathOperator{\BQP}{\mathsf{BQP}}

\DeclareMathOperator{\NP}{\mathsf{NP}}

\DeclareMathOperator{\QMA}{\mathsf{QMA}}

\newcommand{\poly}{\operatorname{poly}}

\newcommand{\id}{\mathrm{I}}

\newcommand{\eps}{\varepsilon}
\newcommand{\expec}{\mathbb{E}}

\newtheorem{theorem}{Theorem}[section]

\newtheorem{lemma}[theorem]{Lemma}

\newtheorem{claim}[theorem]{Claim}

\crefname{assumption}{assumption}{assumptions}

\theoremstyle{definition}

\theoremstyle{remark}
\newtheorem{remark}[theorem]{Remark}

\crefname{property}{Property}{Properties}

\crefname{algocf}{Algorithm}{Algorithms}

\newcommand{\Var}{\mathrm{Var}}

\newcommand{\ketbratwo}[2]{\left| #1 \middle\rangle \middle\langle #2 \right| }

\AtBeginDocument{}

\usepackage[most,breakable]{tcolorbox}
\usepackage{changepage}
\usepackage[final,nopatch=footnote]{microtype}
\newtcolorbox{myframe}[1][]{
  enhanced,
  arc=0pt,
  outer arc=0pt,
  colback=white,
  boxrule=0.8pt,
  #1
}
\newtcolorbox{algbox}{breakable,colback=gray,colframe=gray,standard jigsaw,opacityback=0.15,opacityframe=0.5,boxrule=0.75pt,before upper={\parindent15pt}}

\title{An energy-based uncertainty principle and low-energy state preparation}
\author{Anurag Anshu}
\affil{School of Engineering and Applied Sciences, Harvard University}
\date{\today}
\begin{document}
\captionsetup{width=.9\linewidth}

\maketitle

\begin{abstract}
    Preparing low-energy states of many-body Hamiltonians is a central challenge in quantum computing, quantum complexity, and condensed matter physics. Existing approaches often get trapped in suboptimal states such as high-energy eigenstates or, more generally, low-variance states that resist further energy reduction. In this work, we explore a different perspective: instead of optimizing with respect to a single Hamiltonian, we leverage the fact that many systems admit families of Hamiltonians that share similar low-energy subspaces but differ at higher energies. We show that this redundancy can be turned into an algorithmic resource by establishing an energy-based uncertainty principle, which implies that these Hamiltonians cannot simultaneously admit low-variance states at higher energies. This suggests a simple strategy of alternating energy-lowering steps across such Hamiltonians, which we investigate numerically on several models. We also introduce a sparse variant where the uncertainty principle yields quadratically larger variance at higher energies, leading to more pronounced energy change. Overall, this work suggests a range of open questions at the interface of random matrix theory, local Hamiltonians and low-energy state preparation, aimed at understanding when such approaches are practical and how they can be analyzed rigorously.
\end{abstract}

\section{Introduction}

Approximating the ground energy of a local Hamiltonian is a central problem in quantum many-body physics, quantum complexity theory and a possible candidate for practical quantum advantage \cite{dalzell23book}. It is closely tied with preparing the low energy states themselves, as access to a low energy state directly implies an estimate to the ground energy. However, low energy regime of quantum local Hamiltonians involves highly entangled states, so one is generally restricted to starting with a simple higher energy state (such as a product state) and then improving its energy through a quantum algorithm. 

A challenge with using simple high energy states as a starting point is that such states have very small overlap with the low energy subspace. For example, the energy distribution of any $n$ qubit product state sharply concentrates in the $O(\sqrt{n})$ energy window around the expected energy. This limits significant `cooling' even with non-trivial quantum algorithms such as fixed point amplitude amplification \cite{YLC14}. Variational methods have been proposed to improve the energy \cite{Peruzzo2014VQE, Kandala2017, Aruteetalvariational, Grimsley2019} in practical settings, and they can provably achieve an extensive energy improvement on various families of initial states \cite{AGKS21}. However, as discussed below, these methods usually see no improvement once a low variance state (for example, an eigenstate) of the Hamiltonian is encountered.  It is worth noting that the bottleneck may be fundamental, as complexity theoretic conjectures (such as $\QMA \neq \BQP, \NP \neq \BQP$) prohibit the preparation of low energy states efficiently on a quantum computer. Theoretical results on approximations to ground states have thus been limited to product states or states generated by shallow circuits \cite{GK12, BrandaoH16, BGKK19, GP19, AGKS21, ParekhT21, King2023, ALMPSS25}.  

Our goal in the paper is to find ways to make progress despite getting stuck in a low variance state or a state with negligible mass in the low energy regime. Our starting point is a central insight from quantum Hamiltonian complexity : frustration-free ground states often have multiple non-commuting parent (local) Hamiltonians which agree on the ground state, but seem to have higher degree of non-commutation at higher energies.  For every local Hamiltonian $H= \sum_i h_i$ that is a sum of non-negative hermitian operators (which includes, but is not restricted to, frustration-free models), we define a family of `altered' local Hamiltonians by adding extra penalty terms in the subspace spanned by $h_i$ (see Equation \ref{eq:alteredHloc}). If $H$ has a zero energy ground space, then all the Hamiltonians in the family have the same ground space. Moreover, we show that this family of Hamiltonians satisfies an energy-based uncertainty principle (Theorem \ref{theo:highvar}): states with high energy respect to $H$ have high variance with typical Hamiltonians in the family. 

This gives us a way to avoid getting stuck in a low variance state, by using any energy lowering procedure and then alternating across the local Hamiltonians in the family to maintain a large variance at each step. We formulate the resulting algorithm in Section \ref{sec:var_algorithm} and perform numerical simulations on the AKLT and Heisenberg models. As expected from the above discussion, we find an overall decrease in energy even when we start from eigenstates of the original Hamiltonian (Figure \ref{fig:QMCsim}). Sparse Hamiltonians form an interesting class of Hamiltonians that can be simulated on a quantum computer with appropriate access \cite{Childs2010}, and yet exhibit many features of local Hamiltonians. We construct a sparse version of the `altered' Hamiltonians when $H$ is classical, offering a quadratically larger variance in energy - see Theorem \ref{theo:sparsevar}.

However, further work is needed to better understand the convergence of the algorithm and compare with other known methods to prepare low energy states. Simulations in Section \ref{sec:var_algorithm} have not revealed advantage over prior methods in full generality - likely due to a lack of serious bottlenecks on small system sizes. Indeed, Subsection \ref{subsec:J1J2} does reveal advantage when the ansatz is restricted to be translation invariant (which leads to energy bottlenecks already at $20$ qubits).  At the same time, there are some hints of convergence; Section \ref{subsec:groverdrop} discusses that convergence in sub-exponential time is plausible if certain correlation measure is small or negative. Another missing piece is the mechanism obstructing a fast converging to the low energy regime; which should exist due to fine-grained versions of the aforementioned complexity theoretic conjectures.

\subsection{Related work}

There is a significant body of work understanding the power of variational quantum algorithms \cite{Peruzzo2014VQE,Kandala2017, Aruteetalvariational, Grimsley2019,Cerezo2021VQAReview, AGKS21, HastingsODonnell22, KohYu21, Alvertis2025classicalbenchmarks} for quantum optimization problems and relatedly QAOA \cite{Farhi2014QAOA} for classical optimization problems. The variational algorithms often face the issue of barren plateaus \cite{McClean2018BarrenPlateaus, Cerezo2021VQAReview, Grant2019InitializationStrategy} (which is a significant challenge in the presence of noise \cite{Wang2021NoiseInducedBarrenPlateaus} - we do not consider noise in this paper). A `warm start' is a potential method in mitigating barren plateaus \cite{CGMG26, puig2026warmstartscoldstates} - see Discussion item \ref{discuss:warmstart} for the relevance of warm start in our context.

Alternating minimization across multiple cost functions is a common approach in optimization. In quantum many-body context, White's DMRG method \cite{White92} alternates between various tensors to minimize the energy. While a rigorous proof of the convergence of DMRG does not yet exist, rigorous algorithms inspired by this method are now known \cite{LVV15, ALVV17}.

The use of `noise' for preparing ground states and low energy Gibbs states by dissipation or Gibbs samplers has been a well known recent theme in quantum computing \cite{VWC2009, GilyenS17, PHO21, cubitt2023, DIWB25, FLLC26, STPS26unitaryimaginarytimeevolution, Temme2011, CKBG25, Lin25, RouzeFA25, jiang2024quantummetropolissamplingweak}. Quantum Gibbs samplers and dissipative methods fundamentally rely on the spectral gap of a relevant quantum channel for convergence. In comparison, as discussed in Section \ref{subsec:groverdrop}, the runtime of our approach may be a function of the energy density, possibly scaling exponentially in $\frac{1}{\eps}$ to reach energy below $E_0+\eps n$ (where $E_0$ is the ground energy) if the aforementioned correlation conditions are met. Another point of difference is that there are clear classical analogues of quantum Gibbs sampler and dissipative methods, whereas our approach needs quantum operation even for a classical Hamiltonian due to its reliance on the uncertainty principle.  Inherently quantum operators for a classical problem are also needed in the QAOA \cite{Farhi2014QAOA} and adiabatic algorithm \cite{FGGS20}, although for arguably different reasons.

Reducing energy by imaginary time evolution is another method that is known to be efficiently implementable as long as decay of correlation is maintained \cite{Motta2020}. We currently do not know if decay of correlation can be guaranteed in our recursive steps, so its unclear how to use this method.

\section{Energy-based uncertainty principle}
\label{sec:uncertaintyprinciple}

\subsection{Local Hamiltonian families}

Fix two local Hamiltonians $H_0=\sum_{i=1}^m \Pi_i$, where $\Pi_i$ are projectors supported on a constant number of qubits, and $H= \sum_{i=1}^m h_i$ where for each $i$, $h_i$ is a psd matrices in the subspace projected by $\Pi_i$. We will assume that the non-zero eigenvectors of $h_i$ are constants, which means that $\Omega(1)\Pi_i\preceq h_i\preceq O(1)\Pi_i$. Let $\vec{\phi}={\phi_1, \phi_2, \ldots \phi_m}$ be an array of states, with $\phi_i$ acting on qubits in the support of $\Pi_i$ and $\Pi_i \phi_i = \phi_i \Pi_i= \phi_i$ (this says that $\phi_i$ is a vector in the subspace projected by $\Pi_i$). Define the altered Hamiltonian:
\begin{equation}
\label{eq:alteredHloc}
H_{\vec{\phi}}= H + \Phi, \quad \Phi=\sum_{i=1}^m d_i\phi_i,
\end{equation}
where $d_i=\Tr(\Pi_i)$. Let $D$ be a pair-wise independent distribution over $\vec{\phi}$
and we choose every pair ${\phi_i, \phi_j}$ to form a 2-design. It is known that such $D$ can have $O(nd)$ support size, which leads to $O(nd)$ Hamiltonians. Note that 
\begin{equation}
\label{eq:expectedPhi}
\expec_{\vec{\phi}\sim D} \Phi = \sum_{i=1}^m d_i\frac{\Pi_i}{d_i}= H_0.
\end{equation}
Define the variance of an operator $O$ with respect to a state $\psi$ as the quantity $$\Var_{\psi}(O) := \Tr((O-\Tr(O\psi))^2\psi)= \Tr(O^2\psi) - \Tr(O\psi)^2.$$

The following theorem says that any high energy state of $H$ has high variance with respect to a typical Hamiltonian in Equation \ref{eq:alteredHloc}. 
\begin{theorem}
 \label{theo:highvar}
Suppose $d_i\geq 2$ for all $i$ \footnote{This assumption simply says that there is a `space' to define altered Hamiltonians. We can always reduce to this case by enlarging the Hilbert space per qudit by $1$ and then adding extra penalty terms that prohibit accessing that space.}. For any (possibly mixed) quantum state $\psi$, 
$$\expec_{\vec{\phi} \sim D} \Var_{\psi}(H_{\vec{\phi}}) = \Omega(1) \Tr(H_0\psi) = \Omega(1)\Tr(H\psi).$$
\end{theorem}

\begin{proof}
Consider,
$$\Tr(\psi H_{\vec{\phi}})^2 =  \Tr(\psi H)^2 + \Tr(\psi \Phi)^2 + 2\Tr(\psi H)\Tr(\psi \Phi),$$
and 
$$\Tr(\psi H^2_{\vec{\phi}}) =  \Tr(\psi H^2) + \Tr(\psi \Phi^2) + \Tr(\psi H \Phi) + \Tr(\psi \Phi H).$$
Thus,
\begin{align}
\label{eq:varexpressionHphi}
\Tr(\psi H^2_{\vec{\phi}}) - \Tr(\psi H_{\vec{\phi}})^2& = \Tr(\psi H^2) - \Tr(\psi H)^2 + \Tr(\psi \Phi^2) - \Tr(\psi \Phi)^2\nonumber\\
& + \Tr(\psi H \Phi) + \Tr(\psi \Phi H) -  2\Tr(\psi H)\Tr(\psi \Phi).
\end{align}
Now, taking expectation with $\vec{\phi}$ and using Equation \ref{eq:expectedPhi}, we obtain 
\begin{align}
\label{eq:phivarlowerbound}
&\expec_{\vec{\phi}}\left(\Tr(\psi H^2_{\vec{\phi}}) - \Tr(\psi H_{\vec{\phi}})^2\right) = \nonumber\\
&\Tr(\psi H^2) - \Tr(\psi H)^2 + \expec_{\vec{\phi}}\left(\Tr(\psi \Phi^2) - \Tr(\psi \Phi)^2\right) + \Tr(\psi HH_0) + \Tr(\psi H_0 H) -  2\Tr(\psi H)\Tr(\psi H_0)\nonumber\\
& = \Tr(\psi (H+H_0)^2) - \Tr(\psi (H+H_0))^2 + \expec_{\vec{\phi}}\left(\Tr(\psi \Phi^2) - \Tr(\psi \Phi)^2\right) - \Tr(\psi H_0^2) + \Tr(\psi H_0)^2\nonumber\\
&\geq \expec_{\vec{\phi}}\left(\Tr(\psi \Phi^2) - \Tr(\psi \Phi)^2\right) - \Tr(\psi H_0^2) + \Tr(\psi H_0)^2.
\end{align}
Now, consider the term involving $\Phi$:
\begin{align*}
&\expec_{\vec{\phi}}\left(\Tr(\psi \Phi^2) - \Tr(\psi \Phi)^2\right)  = \expec_{\vec{\phi}}\left(\sum_{i,j=1}^md_id_j\Tr(\phi_i\phi_j\psi) - \sum_{i,j=1}^md_id_j\Tr(\phi_i\psi)\Tr(\phi_j\psi)\right)\\
& =\expec_{\vec{\phi}}\left(\sum_{i=1}^md^2_i\Tr(\phi_i\psi) + \sum_{i\neq j}d_id_j\Tr(\phi_i\phi_j\psi) - \sum_{i=1}^md^2_i\Tr(\phi_i\psi)^2+\sum_{i\neq j}^md_id_j\Tr(\phi_i\psi)\Tr(\phi_j\psi)\right) \\
&=\sum_{i=1}^md_i\Tr(\Pi_i\psi) + \sum_{i\neq j}\Tr(\Pi_i\Pi_j\psi) - \sum_{i\neq j}^m\Tr(\Pi_i\psi)\Tr(\Pi_j\psi)-\expec_{\vec{\phi}}\left(\sum_{i=1}^md^2_i\Tr(\phi_i\psi)^2\right)\\
&=\sum_{i=1}^m(d_i-1)\Tr(\Pi_i\psi) +\sum_{i}^m\Tr(\Pi_i\psi)^2+ \Tr(H_0^2\psi) - \Tr(H_0\psi)^2-\expec_{\vec{\phi}}\left(\sum_{i=1}^md^2_i\Tr(\phi_i\psi)^2\right),
\end{align*}
where in the last equation, we ordered the terms to collect the contributions from $H_0, H_0^2$ separately. Combining back in Equation \ref{eq:phivarlowerbound}, we obtain 
$$\expec_{\vec{\phi}}\left(\Tr(\psi H^2_{\vec{\phi}}) - \Tr(\psi H_{\vec{\phi}})^2\right)\geq \sum_{i=1}^m(d_i-1)\Tr(\Pi_i\psi) +\sum_{i}^m\Tr(\Pi_i\psi)^2-\expec_{\vec{\phi}}\left(\sum_{i=1}^md^2_i\Tr(\phi_i\psi)^2\right).$$

Next, we will use the following facts about 2 design:
$$\expec_{\vec{\phi} \sim D} \phi_i \otimes \phi_i = \frac{\Pi_i\otimes \Pi_i(\id + Swap) \Pi_i\otimes \Pi_i}{d_i(d_i+1)},$$
which allows us to set 
$$\expec_{\vec{\phi}} d_i^2\Tr(\phi_i\psi)^2 = \frac{d_i}{d_i+1}\left(\Tr(\Pi_i\psi)^2+\Tr((\Pi_i\psi\Pi_i\otimes \Pi_i\psi\Pi_i)Swap)\right)= \frac{d_i}{d_i+1}\left(\Tr(\Pi_i\psi_i)^2+\Tr((\Pi_i\psi_i\Pi_i)^2)\right),$$
where $\psi_i=\Tr_{-i}(\psi)$. Plugging in above, we find
\begin{align*}
&\expec_{\vec{\phi}}\left(\Tr(\psi H^2_{\vec{\phi}}) - \Tr(\psi H_{\vec{\phi}})^2\right)\\
&\geq \sum_{i=1}^m(d_i-1)\Tr(\Pi_i\psi_i) +\sum_{i}^m\frac{\Tr(\Pi_i\psi_i)^2}{d_i+1}-\sum_{i=1}^m \frac{d_i}{d_i+1}\Tr((\Pi_i\psi_i\Pi_i)^2)\\
&=\sum_{i=1}^m(d_i-1)\left(\Tr\left(\Pi_i\psi_i\right) - \frac{d_i\Tr((\Pi_i\psi_i\Pi_i)^2)}{d^2_i-1} +\frac{1}{d^2_i-1}\Tr(\Pi_i \psi_i)^2\right)\\
&=\sum_{i=1}^m (d_i-1)\left(\frac{d_i}{d_i+1}\Tr\left(\Pi_i\psi_i\right) + \frac{1}{d_i+1}\Tr\left(\Pi_i\psi_i\right) - \frac{d_i\Tr((\Pi_i\psi_i\Pi_i)^2)}{d^2_i-1} +\frac{1}{d^2_i-1}\Tr(\Pi_i \psi_i)^2\right)\\
&\geq\sum_{i=1}^m (d_i-1)\left(\frac{d_i}{d_i+1}\Tr\left(\Pi_i\psi_i\right) + \frac{1}{d_i+1}\Tr\left(\Pi_i\psi_i\right)^2 - \frac{d_i\Tr((\Pi_i\psi_i\Pi_i)^2)}{d^2_i-1} +\frac{1}{d^2_i-1}\Tr(\Pi_i \psi_i)^2\right)\\
& = \sum_{i=1}^m(d_i-1)\left(\frac{d_i}{d_i+1}\Tr\left(\Pi_i\psi_i\right) + \frac{d_i}{d^2_i-1}\Tr\left(\Pi_i\psi_i\right)^2 - \frac{d_i\Tr((\Pi_i\psi_i\Pi_i)^2)}{d^2_i-1}\right)\\
&\geq \sum_{i=1}^m\frac{(d_i-1)d_i}{(d_i+1)}\Tr\left(\Pi_i\psi_i\right),
\end{align*}
where we used $\Tr((\Pi_i\psi_i\Pi_i)^2)\leq \Tr(\Pi_i\psi_i)^2$ (as $\frac{\Pi_i\psi_i\Pi_i}{\Tr(\Pi_i\psi_i)}$ is a quantum state and hence has purity at most 1). Since $\frac{(d_i-1)d_i}{(d_i+1)}\geq \frac{1}{2}$ for $d_i\geq 2$, we find that 

$$\expec_{\vec{\phi}}\left(\Tr(\psi H^2_{\vec{\phi}}) - \Tr(\psi H_{\vec{\phi}})^2\right) \geq  \frac{1}{2}\sum_i\Tr(\Pi_i\psi).$$ This completes the proof.
\end{proof}

We note that the right-hand side in Theorem \ref{theo:highvar} cannot be improved from linear in energy to super-linear in energy, since product states have $\Theta(n)$ variance as well as energy with respect to quantum Hamiltonians. To improve the lower bound, more general sparse Hamiltonians can be considered - see Appendix \ref{subsec:sparseuncertainty} for details.

\subsection{Formulation relevant to variational circuits}

Consider a local Hamiltonian $H=\sum_i h_i$ with $\|h_i\|\leq 1$, and a basis of local Pauli operators $\{P_1,P_2,\ldots P_M\}$ that can be divided into constant number of groups $Q_1, Q_2, \ldots Q_g$ such that all the Paulis in a group $Q_{\ell}$ have disjoint support. This can be achieved whenever the Paulis are defined on lattices or on low degree graphs.

\begin{algorithm}[h]
\caption{Variational update algorithm}
\label{alg:variational-update}

\KwIn{
A quantum state $\ket{\psi}$ preparable by a circuit $U$ acting on $\ket{0}^{\otimes n}$;
a variational parameter $\theta\in(-1,1)$;
an error parameter $\eps>0$.
}
\KwOut{
A quantum state $\left(\bigotimes_{\ell=1}^ge^{\iota G_{\ell}\theta}\right)\ket{\psi}$, where
$G_{\ell}=\sum_{k\in Q_{\ell}} (\mu_k\pm \eps)P_k$,
$\mu_k=\text{sign}\bigl(\iota\bra{\psi}[P_k,H]\ket{\psi}\bigr)$ and the unitary that prepares $\left(\bigotimes_{\ell=1}^ge^{\iota G_{\ell}\theta}\right)\ket{\psi}$.
}

\For{$k\gets 1$ \KwTo $M$}{
Estimate $\mu_k$ (up to error $\eps$) using $O(\log n/\eps^2)$ copies of $\ket{\psi}$\;
}

For all $\ell \in [g]$, construct $G_{\ell}=\sum_{k\in Q_{\ell}} (\mu_k\pm \eps)\,P_k$ using the estimated $\mu_k$'s\;

Output the quantum state
$\left(\bigotimes_{\ell=1}^ge^{\iota G_{\ell}\theta}\right)\ket{\psi}$, and a classical description of the unitary $V = \left(\bigotimes_{\ell=1}^ge^{\iota G_{\ell}\theta}\right)U$.\;
\end{algorithm}

To explain the choice of $\mu_i$, we consider the case where $\eps=0$ (for $\eps>0$, the energy estimates below get worse by $\pm O(\eps M)$):
\begin{align*}
\bra{\psi}\left(\bigotimes_{\ell=1}^ge^{-\iota G_{\ell}\theta}\right)H \left(\bigotimes_{\ell=1}^ge^{\iota G_{\ell}\theta}\right)\ket{\psi} &= \bra{\psi}H \ket{\psi} - \iota \theta \sum_{\ell=1}^g\bra{\psi} [G_{\ell}, H] \ket{\psi} + O(\theta^2 e^{\theta}M)\\
& = \bra{\psi}H \ket{\psi} - \iota \theta \sum_k \mu_k\bra{\psi} [P_k, H] \ket{\psi} + O(\theta^2 e^{\theta}M).
\end{align*}
Above, the first line bounds the second order estimates using standard tools (see for example, \cite{Grimsley2019, AGKS21}). Now, from the choice of $\mu_k$'s,
$$\bra{\psi}\left(\bigotimes_{\ell=1}^ge^{-\iota G_{\ell}\theta}\right)H \left(\bigotimes_{\ell=1}^ge^{\iota G_{\ell}\theta}\right)\ket{\psi} = \bra{\psi}H \ket{\psi} - \theta \sum_{k} |\bra{\psi} (\iota[P_k, H]) \ket{\psi}| + O(\theta^2 e^{\theta}M).$$
Choosing $\theta = c \frac{\sum_{k} |\bra{\psi} (\iota[P_k, H]) \ket{\psi}|}{M}$, for an appropriate constant $c$ that makes $\theta<1$ ($c$ would depend on the degree of the anticommutation graph of $\{P_1,P_2,\ldots P_M\}$, which we assume to be a constant), we find the energy to be
$$\bra{\psi}H \ket{\psi} -  \Omega(1)\frac{(\sum_{k} |\bra{\psi} (\iota[P_k, H]) \ket{\psi}|)^2}{M}= \bra{\psi}H \ket{\psi} -  \Omega(1)\sum_{k} (\bra{\psi} (\iota[P_k, H]) \ket{\psi})^2,$$
where the last line uses Cauchy-Schwartz. Note that $[P_k, H]$ is only supported on qubits in the neighbourhood of $P_k$. Using cyclicity of trace, we can write
\begin{align}
\label{eq:localvardrop}
\bra{\psi}\left(\bigotimes_{\ell=1}^ge^{-\iota G_{\ell}\theta}\right)H \left(\bigotimes_{\ell=1}^ge^{\iota G_{\ell}\theta}\right)\ket{\psi} = \bra{\psi}H \ket{\psi}  -  \Omega(1)\sum_{k} (\Tr{P_k\iota[\psi, H]})^2. 
\end{align}
 Here, the last expression can be viewed as
a "local variance" of $\ket{\psi}$ with respect to $H$, since the quantity is near zero for any state who local marginals nearly match the marginals of an eigenstate of $H$ (see \cite{lin2024nashstatesversuseigenstates} for more general class of such states).  

Recall the Hamiltonians $H=\sum_{i=1}^m h_i$  and $\Pi_i$ as the subspace that supports $h_i$. Let $S_i$ be the set of qubits on which $\Pi_i$ is supported. The following lemma shows that the quantity is large under some assumptions on the state:
\begin{lemma}
\label{lem:localvar}
Assume that for each $i$, all tensor product Pauli operators in $S_i$ are included in $P_k$'s. Given a state $\psi$, it holds that  
$$\expec_{\vec{\phi}}\sum_{k=1}^M \left(\Tr{P_k\iota[\psi, H_{\vec{\phi}}]}\right)^2 = \frac{\Omega(1)}{m}\left(\sum_{i=1}^m\Tr(\Pi_i\psi)\left\|\frac{\Pi_i\psi_{S_i}\Pi_i}{\Tr(\Pi_i\psi_{S_i})} - \frac{\Pi_i}{d_i}\right\|_2\right)^2.$$
\end{lemma}
Thus, expected `local variance' is large as long as the high energy state does not look maximally mixed in most $\Pi_i$ subspaces. A proof is given in Appendix \ref{append:localvariance}.

\section{Alternating minimization method and simulations}
\label{sec:var_algorithm}

In this section, we formulate the alternating approach for low energy state preparation - see Algorithm \ref{alg:alt-min-variational}. Our formulation pre-samples $r$ randomly chosen altered Hamiltonian and repeatedly tries to reduce the energy of the current state with one of these $r$ altered Hamiltonian. 
\begin{algorithm}[h]
\caption{Alternating-minimization method (uses Algorithm ~\ref{alg:variational-update})}
\label{alg:alt-min-variational}

\KwIn{
An initial state $\ket{\psi_0}=U\ket{0}^{\otimes n}$ specified by a unitary $U$ (typically depth $O(1)$);
initial energy $E=\bra{\psi_0}H\ket{\psi_0}$;
number of steps $L\in\mathbb{N}$; number of altered Hamiltonians $r$.}
\KwOut{
A circuit description of depth $O(L)$ for a state with possibly lower energy than $\psi_0$.
}
Sample $r$ Hamiltonians $H^{(1)}, \ldots H^{(r)}$ iid from the altered family (e.g., $H_{\vec{\phi}}$ for local Hamiltonians, or $H_{T,f}$ for sparse family) \; $U_{\mathrm{current}} \gets U$\;
$i \gets 0$\;
\While{$i < L$}{
        Choose $j\in \{1,2,\ldots r\}$ at random\;
        Run Algorithm~\ref{alg:variational-update} on $\ket{\psi_i}=U_{current}\ket{0}^{\otimes n}$ for the Hamiltonian $H^{(j)}$,
        with $\theta$ chosen appropriately (or searched over). Let $\ket{\psi_{i+1}}$ be the output state.
        Set $U_{\mathrm{current}}$ to be the circuit output from Algorithm~\ref{alg:variational-update}\;
    $i \gets i+1$\;
    }
\end{algorithm}

Algorithm \ref{alg:alt-min-variational} is simulated on various models below, with $r=8$ altered Hamiltonians. In the simulations, we are limited by small system sizes. Thus our goal is to test two important aspects: a) does the alternating method help get past the eigenstate bottleneck? b) does energy minimization with respect to multiple altered Hamiltonians still lead to an overall energy drop with respect to the original Hamiltonian? Simulations indicate positively on both aspects. However, as mentiond in the introduction, the simulations have not revealed any system where Algorithm \ref{alg:alt-min-variational} outperforms existing variational methods.  

We consider the following models:
\begin{itemize}
\item The $1D$ AKLT Hamiltonian \cite{AKLT87}, which is a canonical frustration-free model (Figure \ref{fig:AKLTsim}). 
\item The Heisenberg model with magnetic field $H=\sum_{(i,j)\in E}(X_i X_j + Y_i Y_j + Z_i Z_j) + \sum_i Z_i$ on a square lattice, considered in \cite{Kandala2017} (Figure \ref{fig:QMCsim}). 
\item A degree 6 random graph on $16$ vertices, and the quantum max-cut model $H=\sum_{(i,j) \in E}(X_i X_j + Y_i Y_j + Z_i Z_j)$ (Figure \ref{fig:QMCsim}). 
\item A degree 4 random graph on $16$ vertices and the Max-cut model $H=\sum_{(i,j) \in E}(\id + Z_i Z_j)/2$ (Figure \ref{fig:MC16qub}).
\end{itemize}

\begin{figure}[h!]
\centering
\begin{subfigure}[b]{0.49\textwidth}
  \centering
  \includegraphics[width=\linewidth]{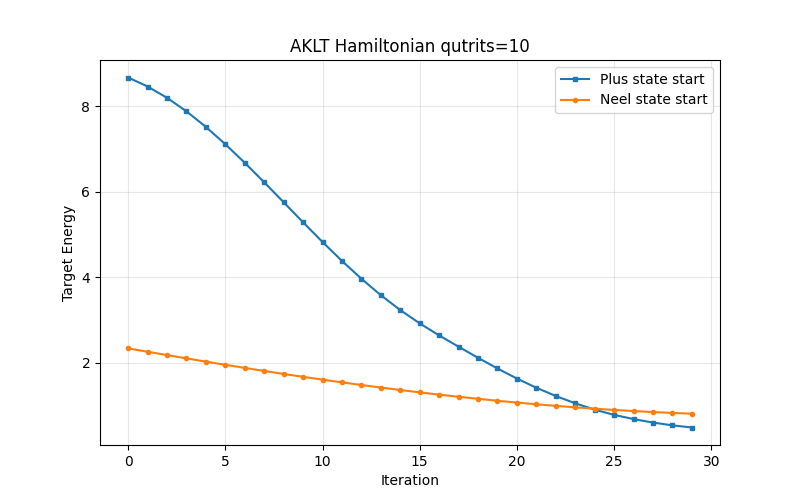}
  \caption{}
\end{subfigure}\hfill
\begin{subfigure}[b]{0.49\textwidth}
  \centering
  \includegraphics[width=\linewidth]{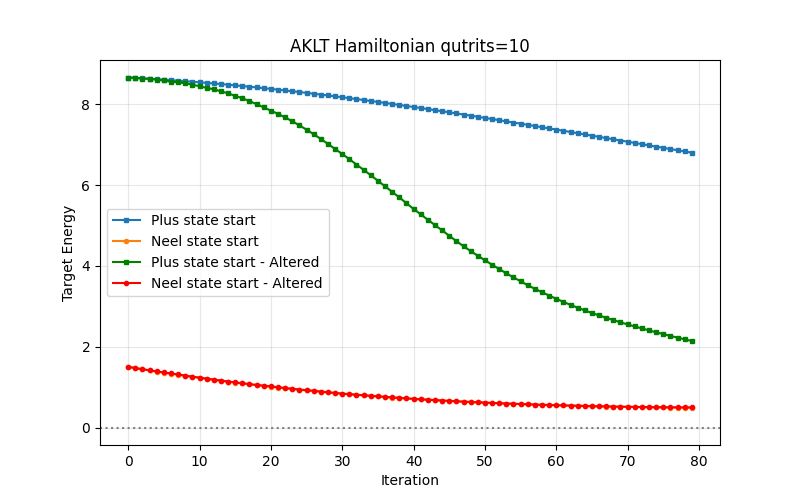}
  \caption{}
\end{subfigure}

\caption{Figure (a) simulates the standard Algorithm \ref{alg:variational-update} for $L=30$ times on 10 qutrit AKLT model, with all 80 non-identity two-qutrit generators in $G$. The algorithm performs very well and reaches the low energy space on both the `plus state' ($\left(\frac{\ket{-1}+\ket{0}+\ket{1}}{\sqrt{3}}\right)^{\otimes 10}$) and the `neel state' ($\ket{1}\otimes\ket{-1}\otimes \ket{1}\ldots$). In Figure (b), we run Algorithm \ref{alg:variational-update} for $L=80$ times with only 15 non-identity spin-1 generators in $G$, again on the plus and neel states. The neel state performance remains good, however the plus state's energy decays slowly. However, the alternating algorithm \ref{alg:alt-min-variational} perfoms much better on the plus state. It performs as well as the standard algorithm on the Neel state.}
\label{fig:AKLTsim}
\end{figure}

\begin{figure}[h!]
\centering

\begin{subfigure}[b]{0.49\textwidth}
  \centering
  \includegraphics[width=\linewidth]{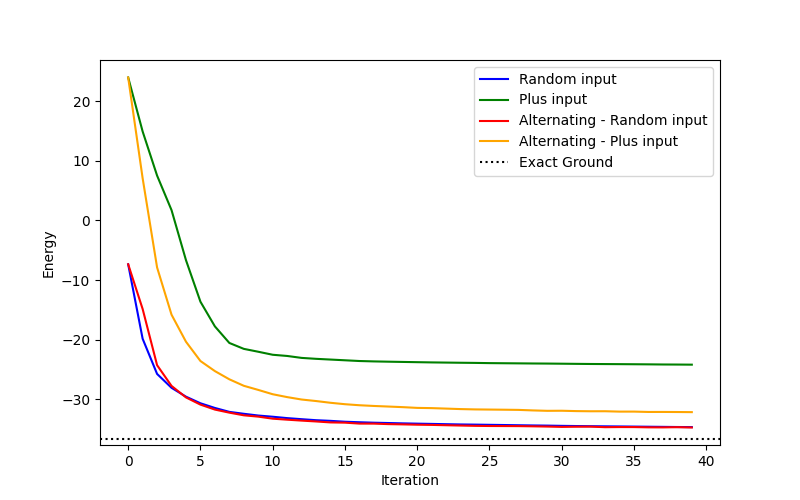}
  \caption{}
\end{subfigure}
\begin{subfigure}[b]{0.49\textwidth}
  \centering
  \includegraphics[width=\linewidth]{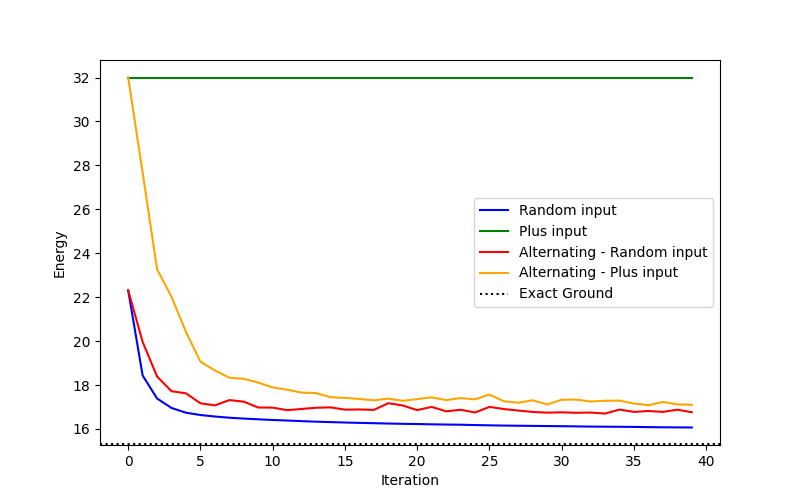}
  \caption{}
\end{subfigure}

\caption{The figures compare Algorithm \ref{alg:variational-update} for $L=20$ steps and Algorithm \ref{alg:alt-min-variational} (for same $L=20$ steps), on plus state $\ket{+}^{\otimes 16}$ and a random product state, on 16 qubits, for two Hamiltonians. Figure (a) depicts the Heisenberg model with magnetic field on a $4 \times 4$ square lattice and Figure (b) depicts the quantum max-cut Hamiltonian on 16 qubits (the same randomly generated graph is used for all the runs in Figure (b)). Note in Figure (b) that the alternating algorithm performs much better than standard for the plus state, which is an eigenstate of the Hamiltonian, but performs slightly worse on the random state.}
\label{fig:QMCsim}
\end{figure}

\begin{figure}
    \centering
    \includegraphics[width=0.65\linewidth]{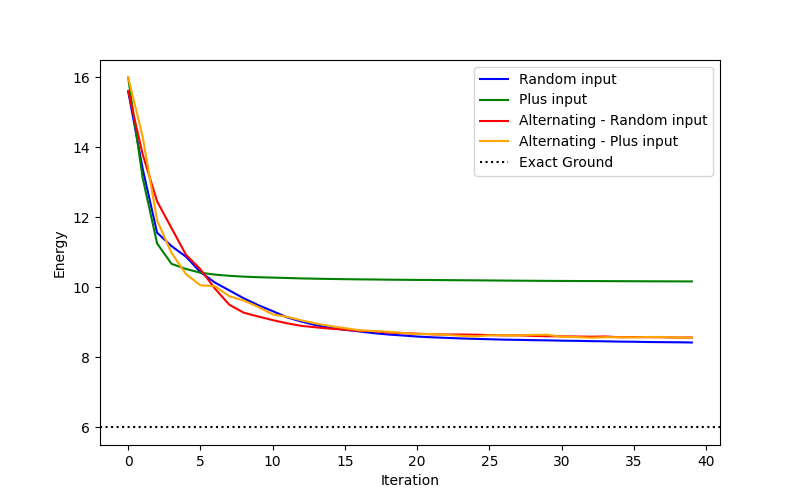}
    \caption{The figure compares Algorithms \ref{alg:variational-update} and \ref{alg:alt-min-variational}, on plus state and random product state, for the Max-cut Hamiltonian on a degree $4$ graph with $16$ vertices.}
    \label{fig:MC16qub}
\end{figure}

\subsection{A (near) translation invariant variational family}
\label{subsec:J1J2}

A challenge with Algorithms \ref{alg:variational-update} and \ref{alg:alt-min-variational} is that they need too many measurements per iteration, which may be expensive even on modern quantum devices. In this section, we explore the possibility of using translation invariant families of these algorithms, and find a numerical advantage with altered Hamiltonians.  

We fix a 2D $n \times m$ periodic lattice and let the edges connect horizontally ($E_H$), vertically ($E_V$) and diagonally (call these edges $E_{D_1}, E_{D_2}$) - see Figure \ref{fig:J1J2altered} (a). The Hamiltonian is
\begin{equation}
\label{eq:J1J2H}
H=J_1\sum_{(i,j)\in E_H\cup E_V}(X_i\otimes X_j + Y_i\otimes Y_j + Z_i \otimes Z_j) + J_2\sum_{(i,j)\in E_{D_1}\cup E_{D_2}}(X_i\otimes X_j + Y_i\otimes Y_j + Z_i \otimes Z_j).
\end{equation}

We consider a variant of Algorithm \ref{alg:variational-update}, where the generator $$G = \sum_{(i,j)\in E_H}\left(\sum_P \mu_{P, H} P\right) + \sum_{(i,j)\in E_V}\left(\sum_P \mu_{P, V} P\right)+\sum_{(i,j)\in E_{D_1}}\left(\sum_P \mu_{P, D_1} P\right) + \sum_{(i,j)\in E_{D_2}}\left(\sum_P \mu_{P, D_2} P\right)$$ is now translation invariant and we optimize only 60 parameter $\mu$'s to minimize the energy $e^{\iota G}\ket{\psi}$. For simulation, we actually use a symmetric trotter formula with $K=4$ steps of trotterization. For a translation invariant $\ket{\psi}$, this breaks translation invariance, but only slightly.  

We run two types of algorithms: first is repeated variational update using the original Hamiltonian; second is switch between variational and altered updates every 10 steps. The altered Hamiltonians are obtained from Equation \ref{eq:J1J2H} by replacing,
on each of the four edge classes, the exchange term
$X\otimes X + Y\otimes Y + Z\otimes Z = 4\Pi_t - 3I$, where $\Pi_t$ is the
projector onto the triplet subspace, by
$4\left(\frac{\Pi_t}{2} + \frac{3}{2}\ket{\phi}\!\bra{\phi}\right) - 3I,$
where $\ket{\phi}$ is a Haar-random triplet vector drawn independently for each
edge class. The same operator is repeated on every edge of its class, so each
altered Hamiltonian remains translation invariant, and it is unbiased:
$\mathbb{E}[H_{\mathrm{alt}}] = H$. Restricting to (near) translation invariant ansatz creates strong energy barriers already at small system sizes, which the altered `bursts' help overcome quickly - see Figure \ref{fig:J1J2altered}(b). 

\begin{figure}
\centering

\begin{subfigure}[b]{0.29\textwidth}
  \centering
  \includegraphics[width=\linewidth]{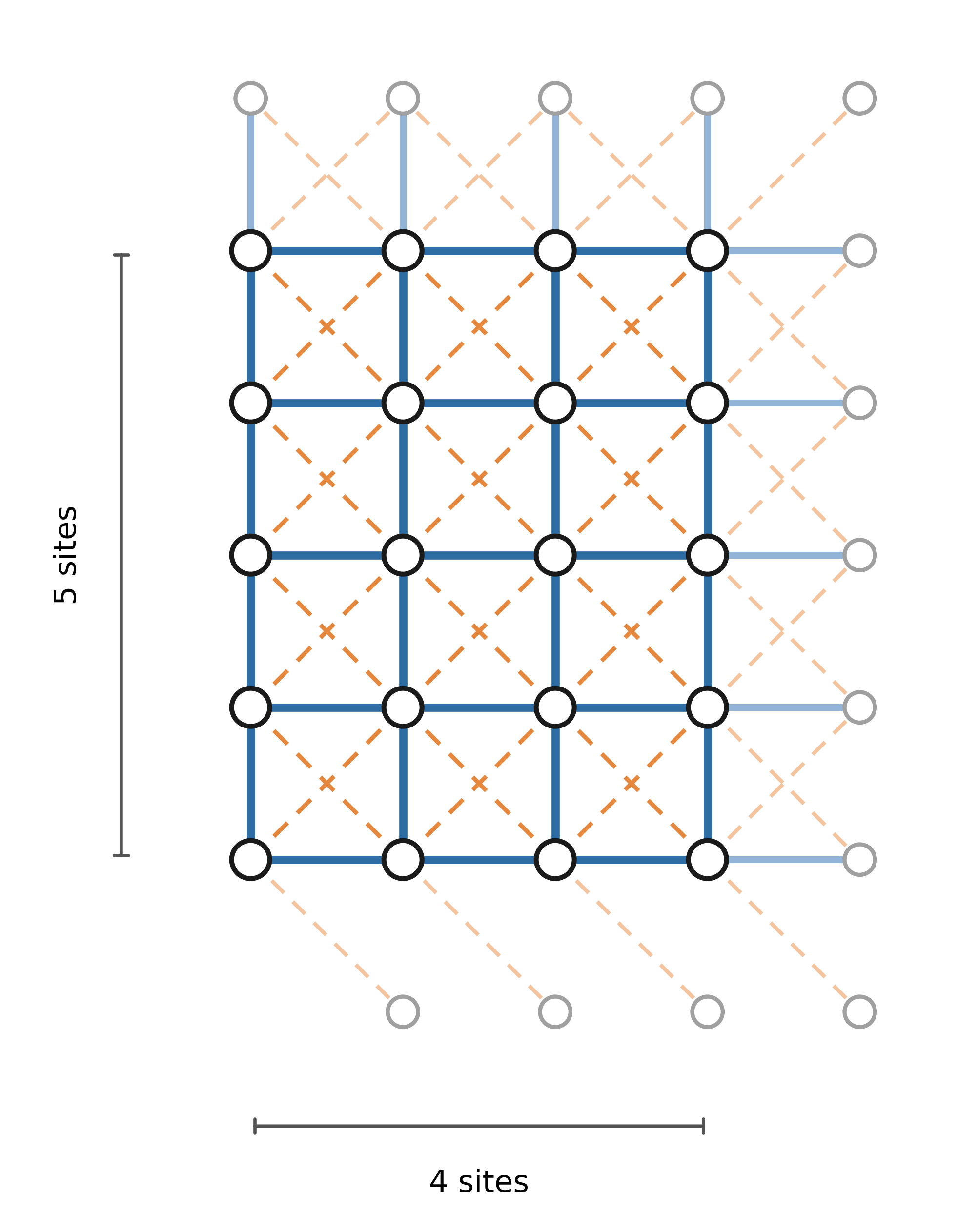}
  \caption{}
\end{subfigure}
\begin{subfigure}[b]{0.59\textwidth}
  \centering
  \includegraphics[width=\linewidth]{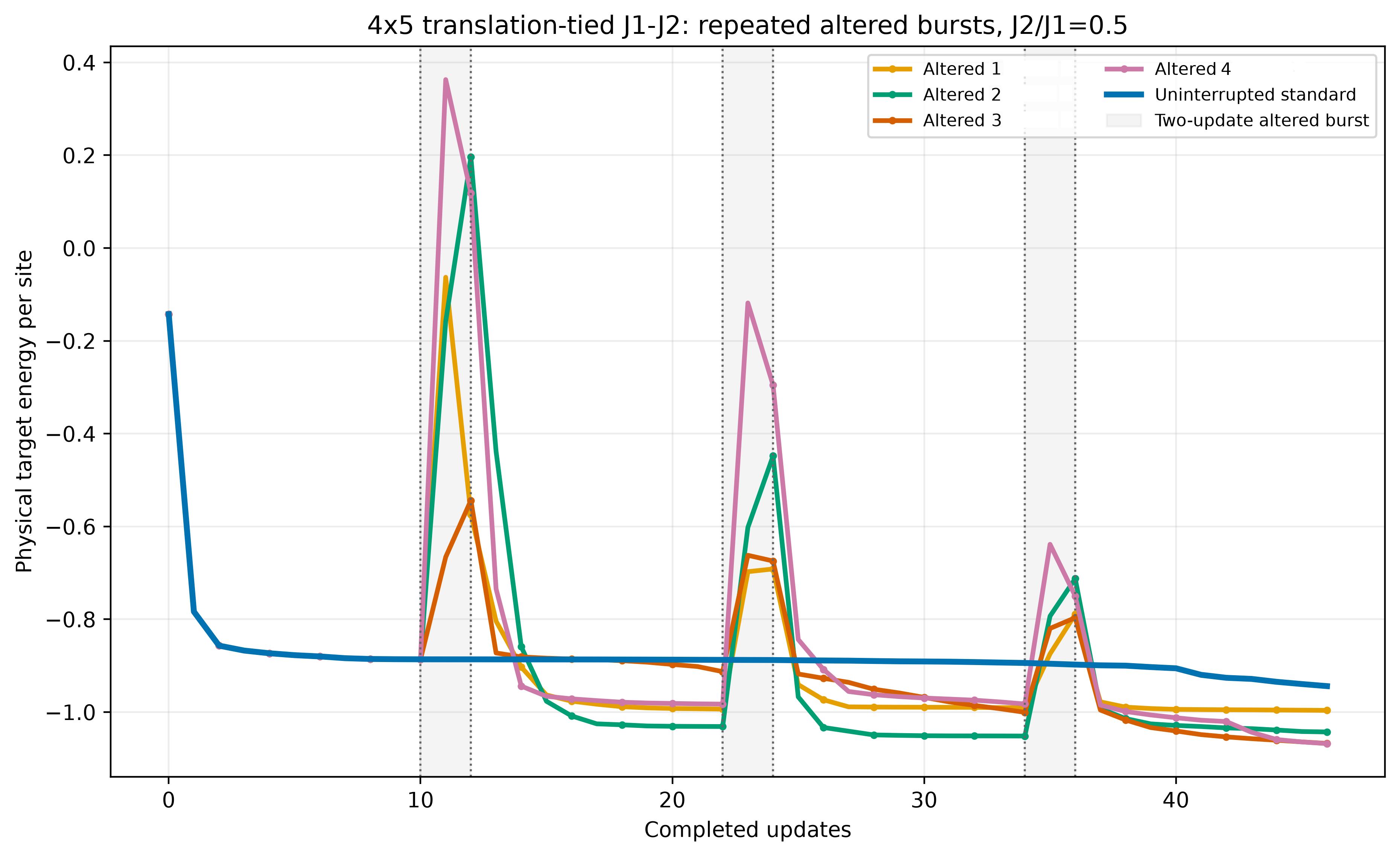}
  \caption{}
\end{subfigure}

    \caption{We run (near) translation invariant version of the variational method with and without assistance of altered Hamiltonians. Figure (a): The Hamiltonian is as in Equation \ref{eq:J1J2H} on $4\times 5$ lattice (with $J_1=1, J_2=0.5$), and the ground energy per site (total ground energy divided by $20$) is $\sim -2.088$ (obtained by the Lanczos method). Figure (b): A random initial product state is chosen as follows: sample two iid single qubit states $\ket{\psi}, \ket{\phi}$, and then place $\ket{\psi}$ on every qubit in the odd column and $\ket{\phi}$ on every qubit in the even column.  The variational approach without using altered Hamiltonian reaches energy per site about $-0.95$ at $46$ updates (blue curve). The curves in other colors correspond to choosing an altered Hamiltonian, running the standard variational updates for 10 iterations followed by altered update for 2 iterations, for 3 cycles (fixed altered Hamiltonian for the run of a fixed color) and then 10 standard updates. We see that the energy per site drops to about $-1.068$ in the best case.}
    \label{fig:J1J2altered}
\end{figure}

\subsection{Adapting Grover-type energy improvement}
\label{subsec:groverdrop}

A natural way to reduce the energy of any state with some energy spread is to use a Grover-type approach. The circuits in this context are more complex than that in Algorithm \ref{alg:variational-update}, but advantage is that the energy improvement relies on a variance like quantity rather than the weaker local-variance like quantity from Equation \ref{eq:localvardrop}. More precisely, given a state $\ket{\psi} = U\ket{0}^{\otimes n}$ that can be prepared by a quantum circuit $U$, let the average energy be $E_{av}=\bra{\psi}H\ket{\psi}$, and let $c$ be such that $\|\Pi^H_{\leq E_{av} - c}\ket{\psi}\|^2 \geq 0.01$. Here, $\Pi^H_{\leq E_{av} - c}$ is the projector onto the eigenspace of $H$ with eigenvalues $\leq E_{av}-c$.  A large $c$ corresponds to large energy spread in typical scenarios\footnote{\label{footnote:spreadvar} We can think of $c$ being similar to $\sqrt{\Var_{H}(\psi)}$; however this is rigorous only if the spread of $\ket{\psi}$ around $E_{av}$ is fairly balanced. In fact, it is possible to come up with states for which $c \approx 0$ despite $\sqrt{\Var_H(\psi)}$ being large, but we do not expect such atypical states to arise in the context of altered Hamiltonians due to significant amount of randomness.}. Using fixed-point amplitude amplification \cite{YLC14}, we can prepare $\frac{\Pi^H_{\leq E_{av} - c}\ket{\psi}}{\|\Pi^H_{\leq E_{av} - c}\ket{\psi}\|}$ using $O(\log\frac{1}{\delta})$ uses of the unitaries $U, U^\dagger, W_{\theta}$ ,   with fidelity $1-\delta$, where $W_{\theta}=\Pi^H_{\leq E_{av}-c} + e^{\iota \theta}(\id - \Pi^H_{\leq E_{av}-c})$ for $\theta \in (0,2\pi)$. Note that $W_{\theta}$ can be well approximated using quantum phase estimation \cite{kitaev1995quantummeasurementsabelianstabilizer}. This improves the energy with respect to $H$ by additive factor of $c$, while increasing the circuit size from $\text{size}(U)$ to $\Theta(\log\frac{1}{\delta})\cdot (\text{size}(U)+\poly(n))$.

Now consider a scheme similar to Algorithm \ref{alg:alt-min-variational}, with variational steps replaced by the above energy reduction scheme for a randomly chosen altered Hamiltonian $H_{alt}$. We will choose a normalization such that $\expec H_{alt} = H$, which is already true for Equation \ref{eq:alteredW} and can be chosen in Equation \ref{eq:alteredHloc} by fixing $H=H_0$ and then dividing by $\frac{1}{2}$. For $T$ iterations starting on a product state, we may set $\delta = \frac{1}{T^2}$ in the fixed point amplitude amplification algorithm, leading to a final state with circuit of size $\Theta(\poly(\log T, n))^T$. 

It is unclear what is the final expected energy as a function of the number of iterations $T$. But we can get some bound depending on a correlation measure defined below. Let the state at the end of iteration $i$ be $\ketbra{\psi_i}$. The expected final energy with respect to $H$ can be written using telescopic sum as 
$$\expec\Tr(\psi_T H) = \Tr(\psi_0 H) + \sum_{i=1}^T\expec(\Tr(\psi_i H) - \Tr(\psi_{i-1}H)).$$
Adding and subtracting the $i$th choice of altered Hamiltonian $H^{(i)}_{alt}$, we obtain
$$\expec\Tr(\psi_T H) = \Tr(\psi_0 H) + \sum_{i=1}^T\expec\left(\Tr(\psi_i H^{(i)}_{alt}) - \Tr(\psi_{i-1}H^{(i)}_{alt})\right)+\sum_{i=1}^T\expec\left(\Tr(\psi_i (H-H^{(i)}_{alt})) - \Tr(\psi_{i-1}(H-H^{(i)}_{alt}))\right).$$
Since $\psi_{i-1}$ is independent of $H^{(i)}_{alt}$, which is chosen using fresh randomness, the last term $\expec\Tr(\psi_{i-1}(H-H^{(i)}_{alt}))$ vanishes using $\expec H^{(i)}_{alt}=H$. Furthermore, Theorem \ref{theo:highvar} suggests (see Footnote \ref{footnote:spreadvar} and discussion item \ref{item:strongervariance}) that in each iteration $i$, $\expec(\Tr(\psi_i H^{(i)}_{alt}) - \Tr(\psi_{i-1}H^{(i)}_{alt})) = -\Theta(\sqrt{\Tr(\psi_i H)})$ (Theorem \ref{theo:sparsevar} respectively suggests that $\expec(\Tr(\psi_i H^{(i)}_{alt}) - \Tr(\psi_{i-1}H^{(i)}_{alt}))= -\Theta(\Tr(\psi_i H)) $). Thus, the energy is expected to drop towards the ground energy if the correlation measure
\begin{equation}
\label{eq:correlation}
\frac{1}{T}\sum_{i=1}^T\expec(\Tr(\psi_i (H-H^{(i)}_{alt})) = \Tr\left(H\left(\frac{1}{T}\sum_{i=1}^T\expec\psi_i\right)\right)-\frac{1}{T}\sum_{i=1}^T\expec\Tr(\psi_i H^{(i)}_{alt})    
\end{equation}
is small or negative. In such a regime, the energy of the state - averaged over all iterations - is not much larger than the energy seen by altered Hamiltonians.   For illustration, suppose this correlation scales as $O(T^{-\alpha}\|H\|) = O(T^{-\alpha}n)$ with $\alpha>0$. Then the largest $T$ for which the algorithm stays above energy $\eps n$ would be determined by the relation $c_1T^{1-\alpha}n - c_2T\sqrt{\eps n}\geq -\|H\| = -c_3 n$, with $c_1, c_2, c_3$ as constants, for the local altered Hamiltonians in Equation \ref{eq:alteredHloc}. This relation is violated when $T= \Theta((n/\eps)^{\frac{1}{2\alpha}})$, implying that the algorithm would reach energy below $\eps n$ in $e^{\Theta((n/\eps)^{\frac{1}{2\alpha}})\log n}$ expected time (subexponential for $\alpha > \frac{1}{2}$). For the sparse altered Hamiltonians (Equation \ref{eq:alteredW}) when $H$ is classical, the condition on $T$ changes to $c_1T^{1-\alpha}n - c_2T\eps n\geq  -c_3 n$, which is violated when $T= \Theta((1/\eps)^{\frac{1}{\alpha}})$. This would give a $\poly(n)^{\Theta((1/\eps)^{\frac{1}{\alpha}})}$ time algorithm to reach energy below $\eps n$.

\section{Discussion}
\label{sec:conclusion}

We have presented a heuristic algorithm aimed at overcoming the barrier of sub-optimal minima (particularly the states with low variance) in low-energy state preparation. Our algorithm is based on an uncertainty principle (Theorems \ref{theo:highvar}, \ref{theo:sparsevar}) for local Hamiltonian families that agree in low energy spectrum, but disagree at higher energies. Several questions are unanswered in this work.

\begin{enumerate}
\item \textbf{Convergence:} The convergence guarantees on the algorithm are largely missing. It is clear that variance lower bound would not suffice for a convergence guarantee and understanding correlation measures such as in Equation \ref{eq:correlation} might be needed. To illustrate the limitations of variance bounds, consider the energy dependent Hamiltonian $H_Q = H + \Tr(H\psi)Q$, where $Q$ is a random local Hamiltonian. This family also exhibits an energy-dependent variance. However, evolving under this Hamiltonian family would be equivalent to random shallow-circuit perturbation when the state gets stuck in a local minima. Such random perturbations would not drive the state towards low energy in a large Hilbert space.

It may be noted that while there are known circuit lower bounds \cite{ABN23, ACKK25} for various spin systems, it is not clear if they limit the convergence of variational Algorithm \ref{alg:alt-min-variational} (which constructs a circuit of depth $O(\sqrt{n})$), since they either apply to low depth circuits or fermionic systems such as the SYK model for which the current formulation does not apply.

\item \textbf{From `expected variance' to stronger measures:}\label{item:strongervariance} The statements in Theorems \ref{theo:highvar} and \ref{theo:sparsevar} concern $\expec_{H_{alt}}\Var_{\psi}(H_{alt})$, for both the families of altered Hamiltonians. However, we would ideally want to show that with high probability,  the variance $\Var_{\psi}(H_{alt})$ is large. This would require calculations involving higher moments.   

Similarly, for understanding convergence, it would be crucial to understand `quantile energies', for example the largest $c$ for which $\|\Pi^{H_{alt}}_{\leq E_{av}-c}\ket{\psi}\|^2\geq 0.01$ as in subsection \ref{subsec:groverdrop}, with $E_{av} =\bra{\psi}H_{alt}\ket{\psi}$. It is plausible that $c$ is large enough in expectation over the randomness in altered Hamiltonians, since a fixed $\psi$ should have a fairly spread energy distribution. However, getting rigorous bounds will require random matrix calculations that go beyond basic moment estimates. 

\item \textbf{Decoupling of energy subspaces:} A different angle at convergence is to understand the interplay between the eigenspaces of $H_{\vec{\phi}}$ and the eigenspaces of $H$. As a concrete set-up, consider $H=H_0=\sum_i \Pi_i$ and recall that $\expec_{\vec{\phi}}H_{\vec{\phi}}/2 = H$ (Equations \ref{eq:alteredHloc} and \ref{eq:expectedPhi}). For some coarse graining parameter $\delta = \frac{1}{\poly(n)}$ and energies $E,F$ consider the subspaces $\Pi^{H_{\vec{\phi}}/2}_{[E-\delta, E+\delta]}$ (eigenspace of $H_{\vec{\phi}}/2$ with energy in the range $[E-\delta, E+\delta]$) and $\Pi^H_{\geq F}$ (eigenspace of $H$ with energy larger or equal to $F$). Is it true that 
\begin{equation}
\label{eq:decouplingq}
\Tr\left(\Pi^H_{\geq F}\expec_{\vec{\phi}}\frac{\Pi^{H_{\vec{\phi}}/2}_{[E-\delta, E+\delta]}}{d_{\vec{\phi}}}\right) < 0.49,
\end{equation}
where $d_{\vec{\phi}}=\Tr(\Pi^{H_{\vec{\phi}}/2}_{[E-\delta, E+\delta]})$, for example when $F=E+\sqrt{E}/20$? In other words, does the `averaged' low energy subspace of altered Hamiltonians have low overlap with higher energy subspace of $H$?

Below, we give a non-rigorous sketch of the utility of Equation \ref{eq:decouplingq}. Given a state $\ket{\psi}$ with extensive energy $E+c$ with respect to $H$, where $c$ represents the deviation in energy w.r.t. altered Hamiltonians as suggested in Theorem \ref{theo:highvar}; let us set $c= \sqrt{E}/10$ for example. Since $\expec_{\vec{\phi}}\bra{\psi}(H_{\vec{\phi}}/2)\ket{\psi} = E+c$. a randomly sampled altered Hamiltonian $H_{\vec{\phi}}/2$ would have $\bra{\psi}(H_{\vec{\phi}}/2)\ket{\psi} \leq E+c+\frac{1}{\poly(n)}$ with probability close to $\frac{1}{2}$. Let $D$ be the distribution over $\vec{\phi}$ conditioned on $\bra{\psi}(H_{\vec{\phi}}/2)\ket{\psi} \leq E+c+\frac{1}{\poly(n)}$, which can be obtained by measuring energy with respect to $\ket{\psi}$ and then rejection sampling. Then, Equation \ref{eq:decouplingq} would imply that $\Tr\left(\Pi^H_{\geq E+\sqrt{E}/20}\expec_{\vec{\phi}\sim D}\frac{\Pi^{H_{\vec{\phi}}/2}_{[E-\delta, E+\delta]}}{d_{\vec{\phi}}}\right) < 0.99$. If we could unitarily  create from $\ket{\psi}$ a state with marginal $\expec_{\vec{\phi}\sim D}\frac{\Pi^{H_{\vec{\phi}}/2}_{[E-\delta, E+\delta]}}{d_{\vec{\phi}}}$, then amplitude amplification would project this state to one in $\Pi^H_{<E+\sqrt{E}/20}$, effectively reducing the energy with respect to $H$. Creating such a state from $\ket{\psi}$ might be possible by using amplitude amplification to prepare $\Pi^{H_{\vec{\phi}}/2}_{[E-\delta, E+\delta]}\ket{\psi}$ and then assuming Eigenstate thermalization hypothesis \cite{Srednicki94} for $H_{\vec{\phi}}$ to employ the expander mixing procedures developed in \cite{AHHN26}.

\item \textbf{Highly frustrated systems:} Simulations in Figures \ref{fig:QMCsim} and \ref{fig:MC16qub} suggest that the algorithm still moves towards the low energy in frustrated systems. However, the method would likely not succeed in extremely frustrated systems such as the SYK model. A key question is to find ways to reduce the frustration in such systems, without changing the ground (or low) energy subspace by much. One well known approach is to use circuit-to-Hamiltonian mappings:   consider the circuit that checks if the energy is below some $a$ or larger than some $b$, and then apply the Feynman-Kitaev mapping \cite{KitaevShenVyalyi2002, aharonov2002quantumnpsurvey}. However, this may not suffice for our purposes as it does not preserve the energy gap $b-a$.

\item \textbf{Cleaner local variance bounds using warm start:} \label{discuss:warmstart} The lower bound in Lemma \ref{lem:localvar} is trivial with high probability if the state is a Haar-random state, since that forces the local marginals to be maximally mixed. However, Haar-random states have nearly the same energy as that of a maximally mixed state (with high probability), which we can always avoid with a warm start. For example, we can start with the lowest-energy efficiently computable product state. Can the condition that energy is extensively below that of the maximally mixed state help in improving the form of the lower bound in Lemma \ref{lem:localvar}?
\end{enumerate}

\subsection{Acknowledgment}

I thank Alex Dalzell, David Gosset, Aram Harrow, Yunchao Liu, Angus Lowe and Harry Zhou for very insightful discussions on the topic, Lillian Zhong for help with proofreading an earlier version of the paper and Yalei Guo for significant help with coding. ChatGPT provided an earlier version of the proof of Lemma \ref{lem:localvar} and ChatGPT 5.6 Sol suggested using the Hamiltonian in Equation \ref{eq:J1J2H}, as well as using altered Hamiltonians for bursts in Section \ref{subsec:J1J2}. ChatGPT and Gemini were used in coding (python codes were verified by the author) and the verification of moment calculations. The GPU codes in Section \ref{subsec:J1J2} were verified by Claude Fable 5. I acknowledge support through the NSF Award Nos. 2238836, 2430375 and QCIS-FF: Quantum Computing \& Information Science Faculty Fellow at Harvard University (NSF 2013303).

\printbibliography

\appendix

\section{Lower bounds on the local variance}
\label{append:localvariance}

Here we provide the proof of Lemma \ref{lem:localvar}.
 
\begin{proof}
Consider
\begin{align*}
&\expec_{\vec{\phi}}\sum_{k=1}^M \left(\Tr{P_k\iota[\psi, H_{\vec{\phi}}]}\right)^2 = \expec_{\vec{\phi}}\sum_{k=1}^M \left(\Tr{P_k\iota[\psi, H]} + \sum_{i=1}^md_i\Tr{P_k\iota[\psi, \phi_i]}\right)^2\\
& = \sum_{k=1}^M \left(\Tr{P_k\iota[\psi, H]}\right)^2+2\sum_{k=1}^M\left(\Tr{P_k\iota[\psi, H]}\right)\left(\expec_{\vec{\phi}}\sum_{i=1}^md_i\Tr{P_k\iota[\psi, \phi_i]}\right)
\\
&+\expec_{\vec{\phi}}\sum_{k=1}^M \left(\sum_{i=1}^md_i\Tr{P_k\iota[\psi, \phi_i]}\right)^2\\
&= \sum_{k=1}^M \left(\Tr{P_k\iota[\psi, H]}\right)^2+2\sum_{k=1}^M\left(\Tr{P_k\iota[\psi, H]}\right)\left(\Tr{P_k\iota[\psi, H_0]}\right)
\\
&+\sum_{k=1}^M \left(\expec_{\vec{\phi}}\sum_{i\neq j}d_id_j\Tr{P_k\iota[\psi, \phi_i]}\Tr{P_k\iota[\psi, \phi_j]} +\expec_{\vec{\phi}}\sum_{i=1}^md^2_i\Tr{P_k\iota[\psi, \phi_i]}^2\right)\\
& = \sum_{k=1}^M \left(\Tr{P_k\iota[\psi, H]}\right)^2+2\sum_{k=1}^M\left(\Tr{P_k\iota[\psi, H]}\right)\left(\Tr{P_k\iota[\psi, H_0]}\right)
\\
&+\sum_{k=1}^M \left(\sum_{i\neq j}\Tr{P_k\iota[\psi, \Pi_i]}\Tr{P_k\iota[\psi, \Pi_j]} +\expec_{\vec{\phi}}\sum_{i=1}^md^2_i\Tr{P_k\iota[\psi, \phi_i]}^2\right)\\
& = \sum_{k=1}^M \left(\Tr{P_k\iota[\psi, H]}\right)^2+2\sum_{k=1}^M\left(\Tr{P_k\iota[\psi, H]}\right)\left(\Tr{P_k\iota[\psi, H_0]}\right) +\sum_{k=1}^M \left(\Tr{P_k\iota[\psi, H_0]}\right)^2
\\
&+\sum_{k=1}^M \left(\expec_{\vec{\phi}}\sum_{i=1}^md^2_i\Tr{P_k\iota[\psi, \phi_i]}^2-\sum_{i=1}^m\Tr{P_k\iota[\psi, \Pi_i]}^2\right)\\
& \geq \sum_{i=1}^m\sum_{k=1}^M \left(d^2_i\expec_{\vec{\phi}}\Tr{P_k\iota[\psi, \phi_i]}^2-\Tr{P_k\iota[\psi, \Pi_i]}^2\right).
\end{align*}
Observe that by the convexity of the square function, $d^2_i\expec_{\vec{\phi}}\Tr{P_k\iota[\psi, \phi_i]}^2\geq \Tr{P_k\iota[\psi, \Pi_i]}^2$. Thus, for each $i$, we can restrict to the Paulis $P_k$ which are in the support $S_i$ of $\Pi_i$ and get a lower bound. Since $\sum_k \Tr(AP_k)^2 = \Omega(2^{|S_i|})\Tr(A_{S_i}^2) = \Omega(1)\Tr(A_{S_i}^2)$, we conclude the lower bound  
\begin{align*}
&\expec_{\vec{\phi}}\sum_{k=1}^m \left(\Tr{P_k\iota[\psi, H_{\vec{\phi}}]}\right)^2\geq \Omega(1)\sum_{i=1}^m \left(d^2_i\expec_{\vec{\phi}}\Tr{(\iota[\psi_{S_i}, \phi_i])^2}-\Tr{(\iota[\psi_{S_i}, \Pi_i])^2}\right).
\end{align*}
Now, 
\begin{align*}
&d_i^2\expec_{\vec{\phi}}\Tr{(\iota[\psi_{S_i}, \phi_i])^2}- \Tr{(\iota[\psi_{S_i}, \Pi_i])^2} = d_i^2\expec_{\vec{\phi}}\Tr{([\psi_{S_i}, \phi_i]\cdot [\phi_i, \psi_{S_i}])} - \Tr{([\psi_{S_i}, \Pi_i]\cdot [\Pi_i,\psi_{S_i}])}\\
&= 2d_i^2\cdot\expec_{\vec{\phi}}\left(\Tr(\phi_i \psi^2_{S_i}) - \Tr(\phi_i \psi_{S_i})^2  \right) - 2\left(\Tr(\Pi_i\psi^2_{S_i})-\Tr((\Pi_i\psi_{S_i}\Pi_i)^2)\right)\\
&=2\left(d_i\Tr(\Pi_i \psi^2_{S_i}) - d_i\frac{\Tr(\Pi_i \psi_{S_i})^2+\Tr((\Pi_i \psi_{S_i}\Pi_i)^2)}{(d_i+1)} -\Tr(\Pi_i\psi^2_{S_i})+\Tr((\Pi_i\psi_{S_i}\Pi_i)^2) \right)\\
& = 2\left((d_i-1)\Tr(\Pi_i \psi^2_{S_i}) + \frac{\Tr((\Pi_i \psi_{S_i}\Pi_i)^2)-d_i\Tr(\Pi_i \psi_{S_i})^2}{(d_i+1)}\right)\\
&=2(d_i-1)\left(\Tr(\Pi_i \psi^2_{S_i}) + \frac{\Tr((\Pi_i \psi_{S_i}\Pi_i)^2)-\frac{1}{d_i}\Tr(\Pi_i \psi_{S_i})^2}{(d^2_i-1)}+\frac{\frac{1}{d_i}\Tr(\Pi_i \psi_{S_i})^2-d_i\Tr(\Pi_i \psi_{S_i})^2}{(d^2_i-1)}\right)\\
&=2(d_i-1)\left(\Tr(\Pi_i \psi^2_{S_i}) - \frac{1}{d_i}\Tr(\Pi_i \psi_{S_i})^2 + \frac{\Tr((\Pi_i \psi_{S_i}\Pi_i)^2)-\frac{1}{d_i}\Tr(\Pi_i \psi_{S_i})^2}{(d^2_i-1)}\right)
\end{align*}
Since $\Tr((\Pi_i \psi_{S_i}\Pi_i)^2) \geq \frac{1}{d_i}\Tr(\Pi_i \psi_{S_i})^2$ and $$\Tr(\Pi_i \psi^2_{S_i})- \frac{1}{d_i}\Tr(\Pi_i \psi_{S_i})^2 = \|\Pi_i \psi_{S_i}(\id-\Pi_i)\|_2^2 + \Tr((\Pi_i \psi_{S_i}\Pi_i)^2)-\frac{1}{d_i}\Tr(\Pi_i \psi_{S_i})^2,$$ the quantity is manifestly positive if $\Pi_i\psi_{S_i}\Pi_i$ is not proportional to the identity operator in $\Pi_i$. Overall, setting $d_i=O(1)$, we conclude
$$\expec_{\vec{\phi}}\sum_{k=1}^M \left(\Tr{P_k\iota[\psi, H_{\vec{\phi}}]}\right)^2 \geq \Omega(1)\sum_{i=1}^m\Tr(\Pi_i\psi)^2\left\|\frac{\Pi_i\psi_{S_i}\Pi_i}{\Tr(\Pi_i\psi_{S_i})} - \frac{\Pi_i}{d_i}\right\|_2^2.$$
Using the convexity of square, this concludes the proof.
\end{proof}

\section{Energy-dependent uncertainty principle with sparse Hamiltonians}
\label{subsec:sparseuncertainty}

A well known fact in Hamiltonian simulation is that sparse Hamiltonians (which generalize local Hamiltonians) can be efficiently simulated on a quantum computer \cite{Childs2010}. An advantage of considering the more general sparse Hamiltonians is that one expects larger variance at higher energies (even for product states). Unfortunately, we do not know how to write this more general family of sparse altered Hamiltonians in the `local representation' of a local Hamiltonian, given the constraint that the zero-energy eigenspace must be preserved. We settle for a definition that is sparse in the eigenbasis of the Hamiltonian. This at least gives a sparse altered family in the computational basis for classical local Hamiltonians.

Consider a Hamiltonian $H$ with eigen-decomposition $H= \sum_{\alpha=1}^{2^n} E_\alpha \ketbra{\xi_{\alpha}}$ (not necessarily in sorted order). For a classical Hamiltonian $H$, we can take $\ket{\xi_{\alpha}}$ to represent the computational basis state corresponding to $\alpha$. Define the altered Hamiltonians  
\begin{equation}
\label{eq:alteredW}
H_{T, g} = G^\dagger G; \quad G = \sum_{\alpha, \beta=1}^{2^n}T_{\alpha, \beta}g_{\alpha,\beta}\ketbratwo{\xi_\alpha}{\xi_\beta},
\end{equation}
with $T_{\alpha, \beta}$ being an indicator variable,  $g_{\alpha, \beta}$ being gaussian with mean $0$ and standard deviation $\sqrt{\frac{E_\beta}{t_{\beta}}}$, and where $t_{\beta} = |\{\alpha: T_{\alpha, \beta =1}\}|$ . With this choice, we get 
\begin{align}
\label{eq:expectedG}
\expec(H_{T,g}) &= \sum_{\beta,\beta'}\expec \left(\sum_\alpha T_{\alpha,\beta}T_{\alpha,\beta'}g_{\alpha,\beta}g_{\alpha,\beta'} \ketbratwo{\xi_{\beta'}}{\xi_{\beta}} \right)\nonumber\\
&= \sum_{\beta} \left(\sum_\alpha T_{\alpha,\beta}\frac{E_\beta}{t_\beta} \ketbratwo{\xi_{\beta}}{\xi_{\beta}} \right)= H.
\end{align}
where we used that $g_{\alpha,\beta}$ has zero mean to set off-diagonal terms to $0$ and used $\expec(g^2_{\alpha,\beta}) = \frac{E_\beta}{t_\beta}$. Another way of viewing Equation \ref{eq:alteredW} is that 
\begin{equation}
    \label{eq:alteredGdifferent}
    H_{T,g} = \sqrt{H}(L^\dagger L)\sqrt{H},
\end{equation}
where $L$ is a $2^n \times 2^n$ random sparse matrix with $L_{\alpha,\beta} = T_{\alpha,\beta} g'_{\alpha,\beta}$ and $g'_{\alpha,\beta}$ being iid gausian with mean $0$ and standard deviation $\sqrt{\frac{1}{t_\beta}}$. Thus, $H_{T,g}$ can just be viewed as a rescaled version of a standard psd random matrix; scaling done by the Hamiltonian itself.

A basic property of the above definition is that any zero energy eigenstate of $H$ is also a zero energy eigenstate of $H_{T,g}$. 
\begin{claim}
    \label{clm:commongs}
Let $\ket{\xi_\alpha}$ be such that $H\ket{\xi_\alpha}=0$ (if it exists). Then $H_{T,g}\ket{\xi_\alpha}=0$.   
\end{claim}
\begin{proof}
Consider $H_{T,g}\ket{\xi_\alpha}=G^\dagger G\ket{\xi_\alpha} = G^\dagger(\sum_{\alpha'} T_{\alpha',\alpha}g_{\alpha',\alpha}\ket{\xi_{\alpha'}})$. Since $E_\alpha=0$, we have $g_{\alpha',\alpha}=0$ with probability $1$ for all $\alpha'$. This shows that $H_{T,g}\ket{\xi_\alpha}=0$.
\end{proof}
We will consider the following choice of $T$, which mimicks the local nature of the alteration in Equation \ref{eq:alteredHloc}.     \begin{equation}
    \label{eq:Tchoiceham}
     T_{\alpha,\beta} = \begin{cases}1 & \text{if }|\text{bin}(\alpha)\oplus \text{bin}(\beta)|=1\\
     0 & \text{otherwise}\end{cases},   
    \end{equation}
where $\text{bin}$ is the binary representation of the integer in its argument. In other words, $T$ connects $\alpha,\beta$ that differ by a bit. We obtain a stronger energy-based uncertainty principle with this choice. 
\begin{theorem}
\label{theo:sparsevar}
    For any state $\psi$, $T$ according to Equation \ref{eq:Tchoiceham} and a continuity assumption that $E_{\beta'}=\Omega(1)E_{\beta}-O(1)$ whenever  $|\text{bin}(\beta)\oplus \text{bin}(\beta')|\leq 2$, we obtain 
     $$\expec_{g}(\Tr(\psi H^2_{T,g}) - \Tr(\psi H_{T,g})^2) = \Omega(1)\cdot \Tr(\psi H^2) - O(1)\cdot \Tr(\psi H) = \Omega(1)\cdot \Tr(\psi H)^2 - O(1)\cdot \Tr(\psi H).$$
\end{theorem}
Qualitatively, this lower bound is stronger than that in Theorem \ref{theo:highvar} for non-zero energies, as the RHS has a quadratic dependence on energy, rather than linear. The continuity conditions are reasonable for all physical systems with a local Hamiltonian. The proof is in Appendix \ref{append:sparseproof}, which is a standard but tedious moment calculation.

\begin{remark}
\label{rem:sparseH}
The following properties of the above definition can be useful.
\begin{enumerate}
\item {\bf Oracle access to $H_{T,g}$ in the classical case: }
The advantage of Equation \ref{eq:Tchoiceham} is that $H_{T,g}$ is a sparse matrix if $H$ is a classical Hamiltonian. Suppose one is able to compute $E_{\alpha}$ for a given $\alpha$ (such as for a classical $H$, which is possible even with oracle access). Then for each $\beta$, one knows the non-zero entries of the $\beta$'th row of $H_{T,f}$ and also for each $\beta',\beta$ one knows the corresponding entry \footnote{Note that this requires access to exponetially many gaussian random variables, but if the algorithm makes $C$ calls to $H_{T,g}$, it suffices to use $C$-wise independent distribution, which requires only $poly(C)$ resources.}.  Thus, one can run Hamiltonian simulation on $H_{T,g}$ and perform energy measurements. 
\item \noindent {\bf Zero ground energy:} We can assume without loss of generality that $E_{min}:=\min\{E_{\alpha}, \alpha\in \{1,2,\ldots 2^n\}\}$ is $0$. This is because we can run the algorithm for values of $E^*$ (guess for $E_{min}$) starting from $0$ to the maximum energy, replacing each $E_{\alpha}$ in the oracle calls with $E_{\alpha}- E^*$. We move on to the next $E^*$ if we ever see a negative $E_{\alpha}- E^*$. And we keep the smallest possible energy observed, corrected by $+E^*$.
\item \textbf{Controlling the largest eigenvalue:} Suppose we know that the ground energy of $H$ is less than some number $B$, for example by preparing a low energy state and measuring its energy. Then in the oracle call, we can truncate all the eigenvectors of $H$ with eigenvalues more than $B$ to $0$, effectively decoupling that subspace from the rest of the algorithm. This can help prevent escape to the high energy regime.  
\item \textbf{Easy sparse family:} It is known \cite{CDBBT24} that some families of sparse Hamiltonians admit fast quantum algorithm, as the maximally mixed state has a good overlap with the low energy subspace (due to the semi-circle law). However, we do not expect the semi-circle to exist for the choice in Equation \ref{eq:Tchoiceham}, as the randomness is significantly reduced at the lower end of the spectrum. 
\end{enumerate}
\end{remark}

\subsection{Proof of variance lower bound for sparse altered Hamiltonians}
\label{append:sparseproof}

\begin{proof}[Proof of Theorem \ref{theo:sparsevar}]
We will prove the statement for pure state $\psi$; the statement for mixed $\psi$ follows by convexity of square function. Let $\ket{\psi}=\sum_\beta a_\beta\ket{\xi_\beta}$. Consider,
\begin{align}
\label{eq:Tgvarlowerbound}
\expec_g\left(\Tr(\psi H^2_{T,g}) - \Tr(\psi H_{T,g})^2\right) = \expec_g\left(\Tr(\psi (G^\dagger G)^2) - \Tr(\psi G^\dagger G)^2\right).
\end{align}
We evaluate each expectation separately. Consider
\begin{align*}
\expec_g\Tr(\psi G^\dagger G)^2 &=  \left(\sum_{\alpha,\beta, \beta'} T_{\alpha, \beta}T_{\alpha, \beta'}a_{\beta}a^*_{\beta'}g_{\alpha,\beta}g_{\alpha,\beta'}\right)^2
=\sum_{\alpha}\expec_g\left(\sum_{\beta, \beta'} T_{\alpha, \beta}T_{\alpha, \beta'}a_{\beta}a^*_{\beta'}g_{\alpha,\beta}g_{\alpha,\beta'}\right)^2\\
&+ \sum_{\alpha\neq \alpha'}\left(\expec_g\sum_{\beta, \beta'} T_{\alpha, \beta}T_{\alpha, \beta'}a_{\beta}a^*_{\beta'}g_{\alpha,\beta}g_{\alpha,\beta'}\right)\cdot \left(\expec_g\sum_{\beta'', \beta'''} T_{\alpha', \beta''}T_{\alpha', \beta''}a_{\beta''}a^*_{\beta'''}g_{\alpha',\beta''}g_{\alpha',\beta'''}\right)\\
&= \sum_{\alpha}\expec_g\left(\sum_{\beta, \beta', \beta'', \beta'''} T_{\alpha, \beta}T_{\alpha, \beta'}T_{\alpha, \beta''}T_{\alpha, \beta'''}a_{\beta}a^*_{\beta'}a_{\beta''}a^*_{\beta'''}g_{\alpha,\beta}g_{\alpha,\beta'}g_{\alpha,\beta''}g_{\alpha,\beta'''}\right)\\
&+ \sum_{\alpha\neq \alpha'}\left(\sum_{\beta} T_{\alpha, \beta}|a_{\beta}|^2\frac{E_\beta}{t_\beta}\right)\cdot \left(\sum_{\beta} T_{\alpha', \beta}|a_{\beta}|^2\frac{E_\beta}{t_\beta}\right),
\end{align*}
where in the first equality we used the independence of expectations\footnote{Thanks to ChatGPT for this observation.} when $\alpha\neq \alpha'$.
Note that
\begin{align*}
&\sum_{\alpha\neq \alpha'}\left(\sum_{\beta} T_{\alpha, \beta}|a_{\beta}|^2\frac{E_\beta}{t_\beta}\right)\cdot \left(\sum_{\beta} T_{\alpha', \beta}|a_{\beta}|^2\frac{E_\beta}{t_\beta}\right)\\
&= \left(\sum_\alpha\sum_{\beta} T_{\alpha, \beta}|a_{\beta}|^2\frac{E_\beta}{t_\beta}\right)\cdot \left(\sum_{\alpha'}\sum_{\beta} T_{\alpha', \beta}|a_{\beta}|^2\frac{E_\beta}{t_\beta}\right) - \sum_\alpha\left(\sum_{\beta} T_{\alpha, \beta}|a_{\beta}|^2\frac{E_\beta}{t_\beta}\right)^2\\
&= \left(\sum_{\beta} |a_{\beta}|^2E_\beta\right)^2 - \sum_\alpha\left(\sum_{\beta} T_{\alpha, \beta}|a_{\beta}|^2\frac{E_\beta}{t_\beta}\right)^2
\end{align*}
We can evaluate the remaining expectation by three pairings: $\beta=\beta'\neq \beta''=\beta'''$, $\beta=\beta''\neq \beta'=\beta'''$, $\beta=\beta'=\beta''=\beta'''$, and obtain
\begin{align*}
&\expec_g\left(\sum_{\beta, \beta', \beta'', \beta'''} T_{\alpha, \beta}T_{\alpha, \beta'}T_{\alpha, \beta''}T_{\alpha, \beta'''}a_{\beta}a^*_{\beta'}a_{\beta''}a^*_{\beta'''}g_{\alpha,\beta}g_{\alpha,\beta'}g_{\alpha,\beta''}g_{\alpha,\beta'''}\right)\\
&=\sum_{\beta\neq \beta''} T_{\alpha, \beta}T_{\alpha, \beta''}|a_{\beta}|^2|a_{\beta''}|^2\frac{E_\beta E_{\beta''}}{t_\beta t_{\beta''}} + \sum_{\beta\neq \beta'} T_{\alpha, \beta}T_{\alpha, \beta'}a^2_{\beta}(a^*_{\beta'})^2\frac{E_\beta E_{\beta'}}{t_\beta t_{\beta'}}+ 3\sum_\beta T_{\alpha,\beta} |a_\beta|^4\frac{E_\beta^2}{t^2_\beta}\\
&=\sum_{\beta,\beta''} T_{\alpha, \beta}T_{\alpha, \beta''}|a_{\beta}|^2|a_{\beta''}|^2\frac{E_\beta E_{\beta''}}{t_\beta t_{\beta''}} + \sum_{\beta, \beta'} T_{\alpha, \beta}T_{\alpha, \beta'}a^2_{\beta}(a^*_{\beta'})^2\frac{E_\beta E_{\beta'}}{t_\beta t_{\beta'}}+ \sum_\beta T_{\alpha,\beta} |a_\beta|^4\frac{E_\beta^2}{t^2_\beta}\\
&=\left(\sum_{\beta} T_{\alpha, \beta}|a_{\beta}|^2\frac{E_\beta }{t_\beta}\right)^2 + \left|\sum_{\beta} T_{\alpha, \beta}a^2_{\beta}\frac{E_\beta}{t_\beta}\right|^2+ \sum_\beta T_{\alpha,\beta} |a_\beta|^4\frac{E_\beta^2}{t^2_\beta},
\end{align*}
where in the first equality, we used $\expec(g^4_{\alpha,\beta}) = 3\frac{E_\beta^2}{t_\beta^2}$. This means that
\begin{align}
\label{eq:Gexpecsquare}
\expec_g\Tr(\psi (G^\dagger G))^2&= \sum_{\alpha}\left(\left(\sum_{\beta} T_{\alpha, \beta}|a_{\beta}|^2\frac{E_\beta }{t_\beta}\right)^2 + \left|\sum_{\beta} T_{\alpha, \beta}a^2_{\beta}\frac{E_\beta}{t_\beta}\right|^2+ \sum_\beta T_{\alpha,\beta} |a_\beta|^4\frac{E_\beta^2}{t^2_\beta}\right)\nonumber\\
&+\left(\sum_{\beta} |a_{\beta}|^2E_\beta\right)^2 - \sum_\alpha\left(\sum_{\beta} T_{\alpha, \beta}|a_{\beta}|^2\frac{E_\beta}{t_\beta}\right)^2\nonumber\\
&= \Tr(\psi H)^2 + \sum_{\beta} |a_{\beta}|^4\frac{E^2_\beta}{t_\beta} + \sum_\alpha\left|\sum_{\beta} T_{\alpha, \beta}a^2_{\beta}\frac{E_\beta}{t_\beta}\right|^2
\end{align}
Similarly,
\begin{align*}
\expec_g\Tr(\psi (G^\dagger G)^2)&= \expec_g\left(\sum_{\beta', \beta, \beta''}a^*_{\beta'}a_\beta\bra{E_{\beta'}}G^\dagger G\ketbra{E_{\beta''}}G^\dagger G\ket{E_\beta}\right)\\
&=\expec_g\left(\sum_{\beta', \beta, \beta'',\alpha,\alpha'}a^*_{\beta'}a_\beta T_{\alpha, \beta'}g_{\alpha,\beta'}T_{\alpha, \beta''}g_{\alpha,\beta''}T_{\alpha', \beta''}g_{\alpha',\beta''}T_{\alpha', \beta}g_{\alpha',\beta}\right).
\end{align*}
Lets consider two cases: $\alpha=\alpha'$ and $\alpha\neq \alpha'$. In the first case, which is similar to earlier calculation, the expectation is non-zero iff $\beta=\beta'$ (since them being different gives two distinct odd degree terms). In the second case, the only non-zero contribution comes when $\beta=\beta''=\beta'$. Collectively, we find 
\begin{align*}
& \expec_g\left(\sum_{\beta', \beta, \beta'',\alpha,\alpha'}a^*_{\beta'}a_\beta T_{\alpha, \beta'}g_{\alpha,\beta'}T_{\alpha, \beta''}g_{\alpha,\beta''}T_{\alpha', \beta''}g_{\alpha',\beta''}T_{\alpha', \beta}g_{\alpha',\beta}\right)\\
& =\expec_g\left(\sum_{\beta, \beta'',\alpha}|a_\beta|^2 T_{\alpha, \beta}T_{\alpha, \beta''}g^2_{\alpha,\beta}g^2_{\alpha,\beta''}\right)+ \expec_g\left(\sum_{\beta,\alpha\neq \alpha'}|a_{\beta}|^2T_{\alpha, \beta}g^2_{\alpha,\beta}T_{\alpha', \beta}g^2_{\alpha',\beta}\right)\\
&= \sum_{\alpha,\beta\neq\beta''} |a_\beta|^2 T_{\alpha,\beta}T_{\alpha,\beta''}\frac{E_\beta E_{\beta''}}{t_\beta t_{\beta''}} +3\sum_{\beta,\alpha} |a_\beta|^2T_{\alpha,\beta}\frac{E^2_\beta}{t^2_\beta} + \sum_{\beta,\alpha\neq \alpha'}|a_\beta|^2T_{\alpha,\beta}T_{\alpha',\beta}\frac{E^2_\beta}{t^2_\beta}\\
&= \sum_{\alpha,\beta,\beta''} |a_\beta|^2 T_{\alpha,\beta}T_{\alpha,\beta''}\frac{E_\beta E_{\beta''}}{t_\beta t_{\beta''}} +\sum_{\beta,\alpha} |a_\beta|^2T_{\alpha,\beta}\frac{E^2_\beta}{t^2_\beta} + \sum_{\beta,\alpha, \alpha'}|a_\beta|^2T_{\alpha,\beta}T_{\alpha',\beta}\frac{E^2_\beta}{t^2_\beta}\\
&= \sum_{\alpha,\beta,\beta''} |a_\beta|^2 T_{\alpha,\beta}T_{\alpha,\beta''}\frac{E_\beta E_{\beta''}}{t_\beta t_{\beta''}} +\sum_\beta |a_\beta|^2\frac{E^2_\beta}{t_\beta} + \sum_{\beta}|a_\beta|^2E^2_\beta.
\end{align*}
From here, we conclude that
\begin{align}
\label{eq:Gsquareexpec}
\expec_g\Tr(\psi(G^\dagger G)^2) &= \Tr(\psi H^2) + \sum_{\beta}|a_\beta|^2\frac{E_\beta^2}{t_\beta} + \sum_{\beta, \beta''}|a_\beta|^2\frac{E_\beta E_{\beta''}}{t_\beta t_{\beta''}}\sum_\alpha T_{\alpha, \beta}T_{\alpha, \beta''}.
\end{align}
Thus, combining Equations \ref{eq:Gsquareexpec} and \ref{eq:Gexpecsquare}, we get 
\begin{align*}
\expec_g\left(\Tr(\psi (G^\dagger G)^2) - \Tr(\psi G^\dagger G)^2\right) &= \left(\Tr(\psi H^2) - \Tr(\psi H)^2\right) + \sum_{\beta}|a_\beta|^2\frac{E_\beta^2}{t_\beta} + \sum_{\beta, \beta''}|a_\beta|^2\frac{E_\beta E_{\beta''}}{t_\beta t_{\beta''}}\sum_\alpha T_{\alpha, \beta}T_{\alpha, \beta''}\\
& - \sum_\beta|a_{\beta}|^4\frac{E^2_\beta}{t_\beta} - \sum_\alpha\left|\sum_{\beta} T_{\alpha, \beta}a^2_{\beta}\frac{E_\beta}{t_\beta}\right|^2\\
&= \left(\Tr(\psi H^2) - \Tr(\psi H)^2\right) + \sum_{\beta}|a_\beta|^2(1-|a_\beta|^2)\frac{E_\beta^2}{t_\beta} \\ 
&+ \sum_{\beta, \beta''}|a_\beta|^2\frac{E_\beta E_{\beta''}}{t_\beta t_{\beta''}}\sum_\alpha T_{\alpha, \beta}T_{\alpha, \beta''} - \sum_{\alpha}\left|\sum_{\beta} T_{\alpha, \beta}a_{\beta}^2\frac{E_\beta}{t_\beta}\right|^2 .
\end{align*}
Plugging back in Equation \ref{eq:Tgvarlowerbound} and dropping the term $\sum_{\beta}|a_\beta|^2(1-|a_\beta|^2)\frac{E_\beta^2}{t_\beta}$,  we get the lower bound
\begin{align*}
&\expec_g\left(\Tr(\psi H_{T,g}^2) - \Tr(\psi H_{T,g})^2\right) \geq \sum_{\beta, \beta''}\left(|a_\beta|^2- \frac{a_\beta^2 (a^*_{\beta''})^2+(a^*_\beta)^2 (a_{\beta''})^2}{2}\right)\frac{E_\beta E_{\beta''}}{t_\beta t_{\beta''}}\sum_\alpha T_{\alpha, \beta}T_{\alpha, \beta''}\\
&= \sum_{\beta, \beta''}(|a_\beta|^2- Re(a_\beta^2 (a^*_{\beta''})^2))\frac{E_\beta E_{\beta''}}{t_\beta t_{\beta''}}\sum_\alpha T_{\alpha, \beta}T_{\alpha, \beta''}\geq \sum_{\beta, \beta''}(|a_\beta|^2- |a_\beta|^2 |a^*_{\beta''}|^2))\frac{E_\beta E_{\beta''}}{t_\beta t_{\beta''}}\sum_\alpha T_{\alpha, \beta}T_{\alpha, \beta''},
\end{align*}
since $\Tr(\psi H^2) - \Tr(\psi H)^2\geq 0$ and $\sum_{\beta}|a_\beta|^2(1-|a_\beta|^2)\frac{E_\beta^2}{t_\beta}\geq 0$. 

Consider the choice of $T$ in Equation \ref{eq:Tchoiceham}. We have $\sum_{\alpha}T_{\alpha,\beta}T_{\alpha,\beta''} = 2$ wherever $|\text{bin}(\beta)\oplus \text{bin}(\beta'')|=2$. Lets restrict to all such $\beta, \beta'$ and use the continuity of $E_\beta$ to set $E_{\beta''}=\Omega(1)E_\beta - O(1)$. Also note that $t_\beta=n$ for all $\beta$. This simplifies
\begin{align*}
\expec_f\left(\Tr(\psi H_{T,f}^2) - \Tr(\psi H_{T,f})^2\right) &\geq  \sum_{\beta, \beta''}|a_\beta|^2(1-|a_{\beta''}|^2)\frac{E_\beta E_{\beta''}}{t_\beta t_{\beta''}}\sum_\alpha T_{\alpha, \beta}T_{\alpha, \beta''}\\
&\geq \sum_{\beta, \beta'':|\text{bin}(\beta)\oplus \text{bin}(\beta'')|=2}|a_\beta|^2(1-|a_{\beta''}|^2)\frac{\Omega(1)\cdot E^2_\beta -O(E_{\beta})}{n^2}\\
& = \sum_{\beta}|a_\beta|^2\left(\sum_{ \beta''::|\text{bin}(\beta)\oplus \text{bin}(\beta'')|=2}(1-|a_{\beta''}|^2)\right)\frac{\Omega(1)\cdot E^2_\beta -O(E_{\beta})}{n^2}\\
& \geq \sum_{\beta}|a_\beta|^2\left({n \choose 2}-1\right)\frac{\Omega(1)\cdot E^2_\beta -O(E_{\beta})}{n^2} \\
&= \Omega(1)\sum_{\beta}|a_\beta|^2E^2_\beta - O(1)\sum_{\beta}|a_\beta|^2E_\beta.
\end{align*}

This completes the proof.
\end{proof}

\end{document}